\documentclass[submission,copyright,creativecommons,sharealike,noncommercial]{eptcs}
\providecommand{\event}{GCM 2020} 

\usepackage[english]{babel}

\usepackage{graphicx}        
\usepackage{subcaption}

\usepackage{amsfonts}
\usepackage{multirow}
\usepackage{tikz}
\usetikzlibrary{patterns,%
								positioning,%
								calc,%
								scopes,%
								decorations.fractals,%
								shapes.geometric,%
								shadows,%
								decorations.pathreplacing,%
								decorations.pathmorphing,%
								decorations.markings,%
								decorations,%
								shapes,%
								arrows,
                                                                automata,petri,
                                                                lindenmayersystems,
                                                                cd
}
\usepackage{xifthen}

\makeatletter
\tikzset{
  column sep/.code=\def\pgfmatrixcolumnsep{\pgf@matrix@xscale*(#1)},
  row sep/.code   =\def\pgfmatrixrowsep{\pgf@matrix@yscale*(#1)},
  matrix xscale/.code=%
    \pgfmathsetmacro\pgf@matrix@xscale{\pgf@matrix@xscale*(#1)},
  matrix yscale/.code=%
    \pgfmathsetmacro\pgf@matrix@yscale{\pgf@matrix@yscale*(#1)},
  matrix scale/.style={/tikz/matrix xscale={#1},/tikz/matrix yscale={#1}}}
\def\pgf@matrix@xscale{1}
\def\pgf@matrix@yscale{1}
\makeatother

\tikzset{every picture/.style={line width=0.3mm},  every state/.style={draw=black,ellipse,minimum width=1em,minimum height=1em,inner sep=2pt,text=black}}

\usepackage{latexsym,amssymb,amsmath,amsfonts}

\makeatletter
\newcommand*{\relrelbarsep}{.386ex}
\newcommand*{\relrelbar}{%
  \mathrel{%
    \mathpalette\@relrelbar\relrelbarsep
  }%
}
\newcommand*{\@relrelbar}[2]{%
  \raise#2\hbox to 0pt{$\m@th#1\relbar$\hss}%
  \lower#2\hbox{$\m@th#1\relbar$}%
}
\providecommand*{\rightrightarrowsfill@}{%
  \arrowfill@\relrelbar\relrelbar\rightrightarrows
}
\providecommand*{\leftleftarrowsfill@}{%
  \arrowfill@\leftleftarrows\relrelbar\relrelbar
}
\providecommand*{\xrightrightarrows}[2][]{%
  \ext@arrow 0359\rightrightarrowsfill@{#1}{#2}%
}
\providecommand*{\xleftleftarrows}[2][]{%
  \ext@arrow 3095\leftleftarrowsfill@{#1}{#2}%
}

\newcommand{\fourrightarrows}{\substack{\xrightrightarrows{}\\ \xrightrightarrows{}}}
\newcommand{\threerightarrows}{\resizebox{2.2ex}{!}{$
    \substack{\longrightarrow\\ \longrightarrow \\ \longrightarrow}
  $}}

\makeatother

\newtheorem{theorem}{Theorem}

\newtheorem{corollary}{Corollary}
\newtheorem{lemma}{Lemma}

\newtheorem{definition}{Definition}
\newtheorem{example}{Example}
\newtheorem{properties}{Properties}

\newtheorem{remark}{Remark}

\newenvironment{proof}[1][Proof]{\begin{trivlist}\item[\hskip \labelsep {\bfseries #1}]}{\end{trivlist}}

\title{Graph Surfing in Reaction Systems\\ from a Categorial Perspective}

\author{Hans-J\"org Kreowski and Aaron Lye
\institute{University of Bremen, Department of Computer Science\\ 
P.O.Box 33 04 40, 28334 Bremen, Germany}
\email{$\{$kreo,lye$\}$@informatik.uni-bremen.de}
}
\def\titlerunning{Graph Surfing in Reaction Systems from a Categorial Perspective}
\def\authorrunning{H.-J. Kreowski \& A. Lye}

\newcommand{\Nat}{\mathbb{N}}
\newcommand{\reaction}[9]{(#2\colon #3 \rightarrow #1,(#4\colon #5 \rightarrow #1, #6\colon #7 \rightarrow #5),#8\colon #9\rightarrow #1)}
\newcommand{\iso}{\cong}
\newcommand{\id}{1}

\newcommand{\binomb}[2]{\begin{bmatrix}#1\\#2\end{bmatrix}}
\newcommand{\scheme}{\mathit{Scm}}

\newcommand{\Ob}{\mathit{Ob}}
\newcommand{\Mor}{\mathit{Mor}}
\newcommand{\cat}[1]{\mathbf{#1}}

\newcommand{\sigmacat}[1]{\Sigma\text{-}\cat{#1}}

\newcommand{\diagramcat}[2]{\cat{#1}^{#2}}

\newcommand{\eiu}{$\mathit{eiu}$}

\newcommand{\typedgraphsdiagramcat}[1]{\diagramcat{Graphs}{\bullet \to \mathit{TG}#1}}

\newcommand{\init}[1]{\mathit{init}_{#1}}
\newcommand{\INIT}{\mathit{INIT}}
\newcommand{\UNION}{\mathit{UNION}}
\newcommand{\union}{\mathit{union}}

\newcommand{\emptymor}[1]{{\mathit{empty_{#1}}}}

\newcommand{\PB}{\mathit{PB}}

\newcommand{\COLIMIT}{\mathit{COLIMIT}}

\newcommand{\MPT}{\mathit{MPT}}
\newcommand{\emptyhypergraph}{(\emptyset, \emptyset, \emptyset_{\emptyset^*}, \emptyset_{\Sigma})}

\newcommand{\en}{\mathit{en}}
\newcommand{\res}{\mathit{res}}

\newcommand{\trans}{\mathit{trans}}

\begin{document}
\maketitle

\begin{abstract}
Graph-based reaction systems were recently introduced as a generalization
of the intensely studied set-based reaction systems.
They deal with
simple edge-labeled directed graphs, and dynamic semantics of
graph-based reaction systems is defined by graph surfing as a novel kind
of graph transformation where, in a single surf step, reactions are
applied to a subgraph of a given background graph yielding a successor
subgraph.
In this paper, we propose
a categorical approach to reaction systems so that a wider spectrum of
data structures becomes available on which reaction systems can be based.
In this way, many types of
graphs, hypergraphs, and graph-like structures are covered.

\end{abstract}

\section{Introduction}
\label{sec:introduction}
Rozenberg and the first author introduced graph surfing in graph-based
reaction systems as a novel kind of graph transformation
in~\cite{Kreowski-Rozenberg:18,Kreowski-Rozenberg:19}.
They consider simple edge-labeled directed graphs.
A graph-based reaction system consists of a finite background graph $B$
and a set of reactions each
of which is a triple $(R,I,P)$ where $R$ and $P$ are subgraphs of $B$, called
reactant and product respectively, and $I=(I_V,I_E)$ is a pair of sets of
vertices and edges of $B$ respectively, called inhibitor. Such a reaction
is enabled on a state $T$ being a subgraph of $B$ if $R$ is subgraph of $T$ and
none of the element of $I_V$ and $I_E$ belongs to $T$.
The latter allows to
forbid edges without forbidding their sources and targets necessarily.
All enabled reactions
are applied to a state in parallel yielding the union of all their
products as successor state. The iterated application of reactions form
trajectories on the set of subgraphs of the background graph – the
metaphorical graph surfing. Before each step, a context graph can be
added to the current state so that the processing becomes interactive.
Graph-based reaction systems generalize the seminal concept of set-based
reaction systems that was introduced by Ehrenfeucht and Rozenberg more
than 12 years ago in~\cite{Ehrenfeucht-Rozenberg:07}
and has been intensely studied since then (see,
e.g.,~\cite{BEMR11,EPR17,FMP14,Sal12a}).
Set-based reaction systems coincide with graph-based
reaction systems the background graphs of which are discrete graphs and
the inhibitor sets are both empty.

In this paper, we advocate a categorical approach to reaction systems by
defining them over categories that provide empty subobjects,
intersections and unions, \eiu-categories for short.
A wide spectrum of
categories of graphs, hypergraphs and graph-like structures fit into the
approach. The categorical framework is tailored in such a way that
reaction systems over an \eiu-category can be defined in close analogy to the
set- and graph-based reaction systems. The ingredients of set- and
graph-based reaction systems are finite sets/graphs, subsets/subgraphs
including the empty set/empty graph, subset/subgraph inclusions,
intersections of
two subsets/subgraphs, and the unions of finite sets of
subsets/subgraphs. As
the categorical counterparts, we use finite objects, subobjects and
subobject inclusions, as they are provided by every category, and we
require a special initial object with monomorphic initial morphisms as
empty subobjects, pullbacks of monomorphisms as intersections and
special colimits as unions in addition.
This paper continues our work on a categorical approach to reaction
systems that started in~\cite{Kreowski-Lye:20} where we tried to identify basic
categorical notions that allow to define reaction systems generalizing the
known set- and graph-based reaction systems.
The \eiu-categories introduced in the present paper are more restrictive,
but cover still all the relevant examples and provide much more useful categorical
machinery.

The paper is organized as follows.
Section~\ref{sec:preliminaries} provides the categorial framework.
In Section~\ref{sec:reaction-systems-cats}, we introduce the notion of reaction systems over \eiu-categories
exemplifying the conception by a reaction system over the category of hypergraphs.
In Section~\ref{sec:further-look-at-the-categorial-framework},
we show that certain diagram categories are \eiu-categories
such that many categories of graphs, hypergraphs and further graph-like
structures turn out to be \eiu-categories and, therefore, can be employed
as base category for reaction systems.
Section~\ref{sec:category-reaction-systems-over-cat} is devoted to the
question how meaningful morphisms between reaction systems over a
category may look like giving a first answer.
This enables us to define a category of reaction systems over an \eiu-category.
Section~\ref{sec:conclusion} concludes the paper.

\section{The Categorial Prerequisites}
\label{sec:preliminaries}
In this section, the categorical prerequisites are provided that allow
us to define reaction systems over a so-called \eiu-category in the next section.
In Subsection~\ref{subsec:categories-subobjects-finiteobjects},
we recall some well-known categorical notions including
subobjects, finite objects, initial objects, pullbacks, and special colimits (cf., e.g.,~\cite{Ehrig-Ehrig-Prange.ea:06,Adamek-Herrlich-Strecker:09,Ehrig-Ermel-Golas-Hermann:15}).
Based on these concepts, we introduce the notion of an
\eiu-category in Subsection~\ref{subsec:emptyobjects-intersection-unions}.

\subsection{Categorial Preliminaries}
\label{subsec:categories-subobjects-finiteobjects}

A \emph{category} $\cat{C} = (\Ob_{\cat{C}}, \Mor_{\cat{C}}, \circ, \id)$ consists of
a class of \emph{objects} $\Ob_{\cat{C}}$,
a set of \emph{morphisms} $\Mor_{\cat{C}}(A,B)$ for each pair of objects $A,B \in \Ob_{\cat{C}}$,
an associative \emph{composition operation} $\circ\colon \Mor_{\cat{C}}(B,C) \times \Mor_{\cat{C}}(A,B) \to \Mor_{\cat{C}}(A,C)$ for each triple of objects $A,B,C \in \Ob_{\cat{C}}$, and,
an \emph{identity} morphism $\id_A \in \Mor_{\cat{C}}(A,A)$ for each object $A \in \Ob_{\cat{C}}$ such that
$f \circ \id_A = f$ and $\id_B \circ f = f$ for each $f \in \Mor_{\cat{C}}(A,B)$ holds.

We may write $f\colon A \to B$ or $A \xrightarrow[]{f} B$ for $f \in \Mor_{\cat{C}}(A,B)$ and
$A \xrightrightarrows[h]{k} B$ for pairs of morphisms with same domain and codomain.
Let $f\colon A \to B$ and $g\colon B \to C$.
We may write $A \xrightarrow[]{f} B \xrightarrow[]{g} C$ instead of $g \circ f$.

A morphism $f\colon A \rightarrow B$ is a \emph{monomorphism} if, for all pairs $C \xrightrightarrows[h]{k} A$ of morphisms, $f \circ h = f \circ k$ implies $h = k$.


A morphism $f\colon A \rightarrow B$ is an \emph{isomorphism} if there exists an \emph{inverse} morphism $f^{-1}\colon B \to A$ with $f^{-1}\circ f = \id_A$ and $f \circ f^{-1} =\id_B$.
Two objects $A,B$ are \emph{isomorphic}, denoted $A\iso B$, if there is an isomorphism $f\colon A \to B$.

%
%
%
%
%
%
%

A \emph{subobject} of $B$ for some $B \in \Ob_{\cat{C}}$ is an equivalence class of the following equivalence of monomorphisms with codomain $B$:
Two monomorphisms $m_1\colon A_1 \to B, m_2\colon A_2 \to B$ are \emph{equivalent}, denoted by $m_1 \iso m_2$, if there is an isomorphism $i\colon A_1 \to A_2$ such that $m_1 = m_2 \circ i$.

To deal with subobjects, we use their elements as representatives. This does not cause any problem because most categorical concepts and constructions are unique up to isomorphism.

Given subobjects $p_1\colon P_1 \to B$ and $p_2\colon P_2 \to B$, a monomorphism $m\colon P_1 \to P_2$ is a \emph{subobject inclusion} from $p_1$ to $p_2$ if $p_1 = p_2 \circ m$, and we may write $p_1 \subseteq p_2$.

An object is \emph{finite} if its set of subobjects is finite.

An object $\INIT \in \Ob_{\cat{C}}$ is an \emph{initial object} if there is exactly one unique morphism $\init{B}\colon \INIT \to B$ for each object $B \in \Ob_{\cat{C}}$.



Let $p_1\colon P_1 \to B, p_2\colon P_2 \to B$ be morphisms with common codomain $B$.
A \emph{pullback} $(\PB(p_1,p_2), p'_1, p'_2)$ of $p_1$ and $p_2$ is defined by
a pullback object $\PB(p_1,p_2)$ and
morphisms $p'_1\colon \PB(p_1,p_2) \to P_1$ and $p'_2\colon \PB(p_1,p_2) \to P_2$ such that
$p_1 \circ p'_1 = p_2 \circ p'_2$
and the following universal property holds:
For each object $Y$ with morphisms $p''_1\colon Y \to P_1$ and $p''_2\colon Y \to P_2$,
such that $p_1 \circ p''_1 = p_2 \circ p''_2$,
there is a unique \emph{universal morphism} $u \colon Y \to \PB(p_1,p_2)$
such that $p'_1 \circ u = p''_1$ and $p'_2 \circ u = p''_2$.
The following diagram illustrates the situation.
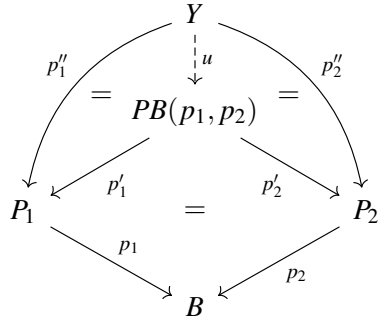
\[
\begin{tikzcd}
& Y \arrow[d, dashed, "u" ] \arrow[ddl, phantom, bend right, "~=" below right] \arrow[ddl, bend right, "p''_1" above left]  \arrow[ddr, bend left, "p''_2" above right] \arrow[ddr, phantom, bend left, "=~" below left]  & \\
& \PB(p_1,p_2)\arrow[dr, "p'_2" below left]\arrow[dl, "p'_1"] \arrow[dd, phantom, "="] & \\
P_1\arrow[dr, "p_1"] && P_2\arrow[dl, "p_2"] \\
& B &
\end{tikzcd}
\]
The dashed arrow indicates that the morphism exists uniquely.

Let $S$ be a set of morphisms with codomain $B$.
Let $\PB(S)$ be the set of all pullbacks $(PB(p_1,p_2),\allowbreak p'_1\colon \PB(p_1,p_2) \to P_1, p'_2\colon \PB(p_1,p_2) \to P_2)$ of $p_1, p_2$ for each pair $(p_1\colon P_1 \to B), (p_2\colon P_2\to B) \in S$ with $p_1 \ne p_2$.
Then an object $\COLIMIT(\PB(S))$ together with a morphism $p''\colon P \to \COLIMIT(\PB(S))$ for each $(p\colon P\to B)\in S$, called \emph{injection}, such that
$p''_1 \circ p'_1 = p_2'' \circ p'_2$ for each pullback $(\PB(p_1,p_2), p'_1, p'_2) \in \PB(S)$
is the \emph{colimit} of $\PB(S)$ if the following universal property holds:
    For each object $X$ together with a morphism $\hat{p} \colon P \to X$ for each $(p\colon P \to B) \in S$
    satisfying $\hat{p}_1\circ p'_1 = \hat{p}_2 \circ p'_2$ for each pullback $(\PB(p_1,p_2), p'_1, p'_2) \in \PB(S)$, there exists a unique \emph{universal morphism} $m\colon \COLIMIT(\PB(S)) \to X$ such that $m\circ p'' = \hat{p}$ for each $p \in S$.

According to the definition, the following holds for three special cases of this colimit.
\begin{enumerate}
\item
$\COLIMIT(\PB(\emptyset))=\INIT$.

\item
$\COLIMIT(\PB(\{p\}))=P$ for each subobject $p\colon P \to B$.

\item
Given two subobjects $p_i \colon P_i \to B, i=1,2$, then
$\COLIMIT(\PB(\{p_1,p_2\}))$ together with the injections
$p''_i:P_i \to \COLIMIT(\PB(\{p_1,p_2\}))$
is the pushout of the pullback
$(\PB(p_1,p_2),p'_1,p'_2)$.
\end{enumerate}

It may be noted that the universal
property of the colimit yields a universal morphism\linebreak $m\colon\COLIMIT(\PB(S)) \to B$
with $m\circ p''=p$ for each $p \in S$.

The following diagram illustrates the situation for three subobjects.
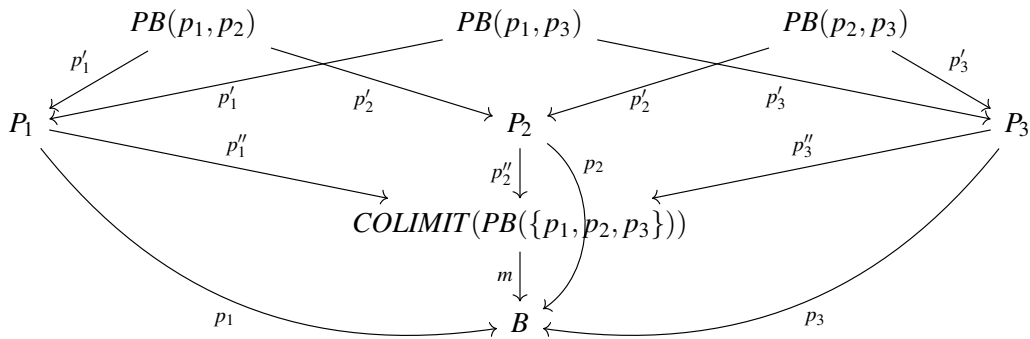
\[
\begin{tikzcd}
& \PB(p_1,p_2) \arrow[dl, "p'_1" above left] \arrow[dr, "p'_2" below left] & \PB(p_1,p_3) \arrow[dll, "p'_1" below left] \arrow[drr, "p'_3" below left] & \PB(p_2,p_3) \arrow[dl, "p'_2" below left] \arrow[dr, "p'_3" above right] &\\
P_1 \arrow[ddrr, bend right, "p_1" below left] \arrow[drr, "p_1''" above right] & & P_2 \arrow[dd, bend left = 55, "p_2" near start]  \arrow[d, "p_2''" left] & & P_3 \arrow[ddll, bend left, "p_3" below right] \arrow[dll, "p_3''" above left]\\ %
& & \COLIMIT(\PB(\{p_1,p_2,p_3\})) \arrow[d, "m" left] & & \\ 
& & B & &
\end{tikzcd}
\]

\subsection{Empty Subobjects, Intersections and Unions}
\label{subsec:emptyobjects-intersection-unions}
Using the notions of the previous subsection, we can now define the
class of categories that are considered in this paper.

A category $\cat{C}$ is an \emph{\eiu-category} if $\cat{C}$ has
\begin{enumerate}
\item an initial object $\INIT$, and
\item for every finite object $B$, pullbacks of the subobjects of $B$, as well as
\item colimits of the sets of all pairwise pullbacks of sets of subobjects of every finite object $B$
\end{enumerate}
subject to the following conditions:
\begin{enumerate}
\item
$\INIT$ has only itself as subobject and the initial morphism into $B$ is a
monomorphism, and
\item
the universal morphism from $\COLIMIT(\mathit{PB}(S))$ into $B$
for every set $S$ of subobjects of $B$ is a monomorphism.
\end{enumerate}
We use the following notions and notations for \eiu-categories and every
of its finite objects $B$.
\begin{enumerate}
\item \label{crs-general-assumption-item1}
The subobject represented by the initial morphism into $B$ is called
\emph{empty subobject} of $B$ and denoted by $\emptymor{B}\colon \INIT \to B$.

\item \label{crs-general-assumption-item2}
As pullbacks are stable under monomorphisms, the pullback morphisms
$p'_i\colon \mathit{PB}(p_1,p_2) \to P_i$ of two subobjects $p_i \colon P_i \to B$ for $i=1,2$ are monomorphisms.
Further, because monomorphisms are closed under composition,
$p'_1\circ p_1=p'_2\circ p_2$ represents a subobject of $B$ called
\emph{intersection} of $p_1$ and $p_2$ which is denoted by $p_1\cap p_2 \colon P_1\cap P_2 \to B$.

\item \label{crs-general-assumption-item3}
Given a set $S$ of subobjects of $B$, the universal morphism from\linebreak
$\COLIMIT(\mathit{PB}(S))$ into $B$ represents a subobject of $B$
called \emph{union} of $S$
which is denoted by $\union(S) \colon \UNION(S) \to B$.
We may write $p_1 \cup p_2$ for the binary (effective) $\union(\{p_1,p_2\})$.
\end{enumerate}

Empty subobjects, intersections and unions have some useful properties
(cf. Remarks~\ref{remark:properties-rs} and~\ref{remark:properties-interactive-processes} in the next section).
The initials e, i, and u of the three
concepts are used to name the category.

\begin{properties}
\label{properties:empty-objects-intersection-union}
Let $B$ be a finite object.
\begin{enumerate}
\item
\label{item:intersection-new-1}
Let $p\colon P \to B$ and $p_0\colon P_0 \to B$ be subobjects of $B$ with $p_0 \subseteq p$. Then
\begin{enumerate}
\item
\label{item:intersection-new-1a}
$p \cap p_0 = p_0$,

\item
\label{item:union-7}
 $p \cup p_0 = p$.
\end{enumerate} 
In particular,
$p \cap \emptymor{B} = \emptymor{B}$ and
$p\cup \emptymor{B} = p$.


\item
\label{item:union-6}
Let $S$ be a set of subobjects of $B$. Then
$\union(S \cup \{\emptymor{B}\}) = \union(S)$.

\item
\label{item:union-8}
Let $S_0$ and $S$ be sets of subobjects of $B$ with $S_0 \subseteq S$. Then $\union(S_0) \subseteq \union(S)$.


\end{enumerate}
\end{properties}

\begin{proof}
1.
$p_0 \subseteq p$ means that there is a monomorphism $m\colon P_0 \to P$ with $p \circ m = p_0$.
Using this equation, it is easy to show that $(P_0, \id_{P_0}, m)$ is a pullback of $p_0$ and $p$ and $(P, m, \id_P)$ is a pushout of $\id_{P_0}$ and $m$.
As pullbacks and pushouts are unique up to isomorphisms, one gets $p \cap p_0 = p_0$ and $p \cup p_0 = p$ for the represented subobjects.
This holds for $p_0 = \emptymor{B}$, in particular.
The following diagrams illustrate the situation.
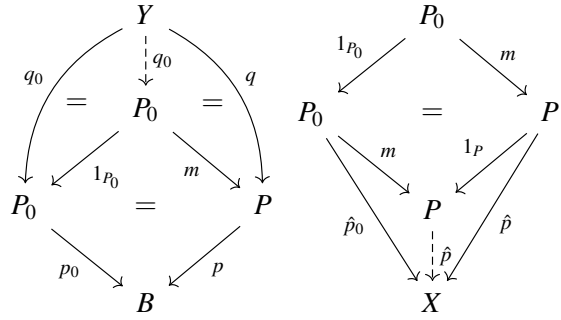
\[
\begin{tikzcd}
& Y \arrow[d, dashed, "q_0" ] \arrow[ddl, bend right, "q_0" above left] \arrow[ddl, phantom, bend right, "~=" below right] \arrow[ddr, bend left, "q" above right] \arrow[ddr, phantom, bend left, "=~" below left] & \\
& P_0 \arrow[dr, "m" below left]\arrow[dl, "\id_{P_0}"] \arrow[dd, phantom, "="] & \\
P_0\arrow[dr, "p_0" below left] && P\arrow[dl, "p"] \\
& B &
\end{tikzcd}
\begin{tikzcd}
&P_0 \arrow[dl, "\id_{P_0}" above left] \arrow[dr, "m" above right] & \\
P_0 \arrow[rr, phantom, "="] \arrow[ddr, "\hat{p}_0" below left] \arrow[dr, "m"] & & P \arrow[dl, "\id_P" above left] \arrow[ddl, "\hat{p}"]  \\
& P \arrow[d, dashed, "\hat{p}"] & \\
& X &
\end{tikzcd}
\]

2.
If $\emptymor{B} \in S$, then $S \cup \{ \emptymor{B} \} = S$ so that the statement holds in this case.

Consider now $S$ with $\emptymor{B} \notin S$.
By definition, $\union(S)\colon \UNION(S) \to B$ is accompanied with a monomorphism $p''\colon P \to \UNION(S)$ for each $(p\colon P \to B) \in S$ such that $p = \union(S) \circ p''$ and, for each pair $(p\colon P \to B), (\overline{p}, \overline{P} \to B) \in S$ with a pullback
$(P \cap \overline{P}, p'\colon P \cap \overline{P} \to P, \overline{p}'\colon P \cap \overline{P} \to \overline{P})$
of $p$ and $\overline{p}$, $p'' \circ p' = \overline{p}'' \circ \overline{p}'$.
Now one can add $\emptymor{B}$ to $S$ and choose $\emptymor{\UNION(S)} \colon \INIT \to \UNION(S)$ as monomorphism corresponding to $\emptymor{B}$.
As the initial morphism is unique, one gets $\emptymor{B} = \union(S) \circ \emptymor{\UNION(S)}$ and
$p'' \circ \emptymor{P} = \emptymor{\UNION(S)} = \emptymor{\UNION(S)} \circ \id_{\INIT}$.
As pointed out in Point~\ref{item:intersection-new-1},
$(P \cap \INIT, \emptymor{P}, \id_{\INIT})$ is a pullback of $p$ and $\emptymor{B}$.
Altogether, this means that $\union(S)$ with morphisms $p''$ plus $\emptymor{\UNION(S)}$ equalizes all pullbacks in $\PB(S \cup \{ \emptymor{B} \})$.
Moreover, one can show that also the universal property of $\union(S \cup \{ \emptymor{B} \})$ is satisfied.
Let $X$ be an object with a morphism $\widehat{p} \colon P \to X$ for each $p \colon P \to B$
plus the only initial morphism $\emptymor{X} \colon \INIT \to X$ such that all pullbacks
in $\PB(S \cup \{ \emptymor{B} \})$ are equalized, i.e.,
($*$) $\widehat{p} \circ p' = \widehat{\overline{p}} \circ \overline{p}'$ for each $(P \cap P', p', \overline{p}') \in \PB(S)$
and $\widehat{p} \circ \emptymor{P} = \emptymor{X} \circ \id_{\INIT}$ for each
pullback $(P \cap \INIT, \emptymor{P}, \id_{\INIT})$ for $p\in S$ and $\emptymor{B}$.
Because \mbox{of ($*$),} the universal property of $\union(S)$ induces a morphism $m\colon \UNION(S) \to X$ with
$\widehat{p} = m \circ p''$ for all $p\in S$.
Moreover, the initiality of $\INIT$ yields $\emptymor{X} = m \circ \emptymor{\UNION(S)}$.
Summarizing, $\union(S)$ with the morphisms $p''$ plus $\emptymor{\UNION(S)}$ has
the property of $\union(S) \cup \{ \emptymor{B} \}$ so that they are equal as subobjects.
The situation is depicted in the following diagram.
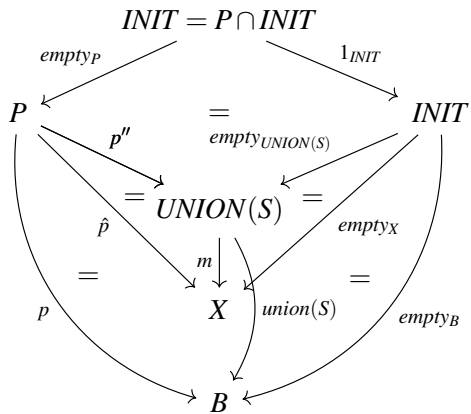
\[
\begin{tikzcd}
& \INIT = P \cap \INIT \arrow[dl, "\emptymor{P}" above left] \arrow[dr, "\id_{\INIT}"] & \\
P \arrow[rr, phantom, "="] \arrow[dddr, bend right = 40, "p" below left] \arrow[dr, "p''"] \arrow[dddr, phantom, bend right = 30, "=" above right] \arrow[dr, "p''"] \arrow[ddr, "\hat{p}" below left] \arrow[ddr, phantom, "=" above right] & & \INIT \arrow[dl, "\emptymor{\UNION(S)}" above left] \arrow[ddl, "\emptymor{X}"] \arrow[ddl, phantom, "=" above left] \arrow[dddl, bend left = 40, "\emptymor{B}"] \arrow[dddl, phantom, bend left = 30, "=" above left]\\
& \UNION(S) \arrow[d, "m" left] \arrow[dd, bend left = 30, "\union(S)"] & \\
& X & \\
& B &
\end{tikzcd}
\]

3.
Using the notation of Point~\ref{item:union-6}, $\union(S)$ with the morphisms $p''$ for $p\in S$
equalizes all pullbacks in $\PB(S)$ and, in particular, all in $\PB(S_0)$ as $S_0 \subseteq S$.
Therefore, using the universal property of $\union(S_0)$, there is a morphism $m\colon \UNION(S_0) \to \UNION(S)$ with $\union(S) \circ m = \union(S_0)$.
As $\union(S_0)$ is a monomorphism, $m$ is a monomorphism proving $\union(S_0) \subseteq \union(S)$.
\end{proof}

\begin{example}
First of all, the category $\cat{Sets}$ with sets as objects and mappings as morphisms is an \eiu-category.
This follows from well-known set-theoretic and categorial properties.
The monomorphisms are the injective mappings.
Two of them with common codomain are equivalent if they have the same image.
Therefore, there is a one-to-one correspondence between subobjects of a set and its subsets,
and subobjects can be represented by the inclusions of subsets.
In particular, the finite sets are the finite objects.
The empty set $\emptyset$ is the initial object.
It has only itself as subset, and the initial morphism $\emptyset_B\colon \emptyset \to B$ is injective for every set $B$ so that $\emptyset_B$ is the empty subobject of $B$.
Given two subsets $P_1$ and $P_2$ of a set $B$, their set-theoretic intersection $P_1 \cap P_2$ together with the inclusion into $P_1$ and $P_2$ respectively is a pullback over the inclusions of $P_1$ and $P_2$ into $B$ and, therefore, the categorial intersection.
Moreover, let $S$ be a set of subsets of a set $B$.
Then the set-theoretic union $\bigcup\limits_{p \in S} P$ is the smallest subset of $B$ that contains each $P \in S$.
If $X$ is a set and $q_P\colon P \to X$ is a mapping for each $P \in S$ such $q_{P_1}$ and $q_{P_2}$ are equal on the intersection $P_1 \cap P_2$ for every pair $P_1,P_2 \in S$, then $m\colon \bigcup\limits_{P \in S} P \to X$ given by $m(y) = q_P(y)$ for $y \in P, P \in S$ is a mapping.
This proves that the inclusions of $\bigcup\limits_{P \in S} P$ into $B$ has the universal property required of $\union(S)$ so that the set-theoretic union turns out to represent the categorial union.

Based on $\cat{Sets}$, many further \eiu-categories can be derived (cf. Section~\ref{sec:further-look-at-the-categorial-framework}).
As a first example of this kind we consider the category $\sigmacat{Hypergraphs}$.
Its objects are $\Sigma$-hypergraphs and its morphisms are $\Sigma$-hypergraph morphisms defined as follows.
A $\Sigma$-\emph{hypergraph} $H = (V,E,att,l)$ over a given set $\Sigma$ of \emph{labels} is a system consisting of a set $V$ of \emph{vertices}, a set $E$ of \emph{hyperedges}, an \emph{attachment} mapping $att\colon E \to V^*$ (assigning a string of attachment vertices to each hyperedge) and a \emph{labeling} mapping $l\colon E \to \Sigma$.
The components of $H = (V,E,att,l)$ may also be denoted by $V_H$, $E_H$, $att_H$, and $l_H$ respectively.
The length of the attachment is called \emph{type}.
A \emph{hypergraph morphism} $f$ from $H = (V,E,att,l)$ to $H' = (V',E',att',l')$ is a pair $(f_V\colon V \to V', f_E\colon E\to E')$ of two mappings
such that
$f^*_V \circ att = att' \circ f_E$
and $l = l' \circ f_E$,
where $V^*$ is the set of all string over $V$ and $f_V^*\colon V^* \to V'^*$ is the canonical extension of $f_V$ to strings defined by $f^*_V(v_1\cdots v_n) = f_V(v_1) \cdots f_V(v_n)$ for all $v_1\cdots v_n \in V^*$.
$H$ is a sub-$\Sigma$-hypergraph of $H'$ if $V_H \subseteq V_{H'}, E_H \subseteq E_{H'}$
and the pair of inclusions $in=(in_V,in_E)$ is a hypergraph morphism.

It is not difficult to see that all the ingredients of \eiu-categories can be carried over from $\cat{Sets}$ to $\sigmacat{Hypergraphs}$ componentwise for vertices and hyperedges.
The monomorphisms are the pairs of injective mappings so that sub-$\Sigma$-hypergraphs correspond to subobjects, and finiteness is given by finite set components.
The empty $\Sigma$-hypergraph $\MPT = \emptyhypergraph$ is initial so that
$\emptymor{B} \colon \MPT \to B$ given by
$\emptyset_V \colon \emptyset \to V_B$
and
$\emptyset_E \colon \emptyset \to E_B$
is the empty subobject of each $\Sigma$-hypergraph $B$.
Analogously, intersection and union can be constructed componentwise.

For $p_i\colon P_1 \to B, i=1,2$
we have
$p_1 \cap p_2\colon P_1 \cap P_2 \to B$ with 
$P_1 \cap P_2 = (V_{P_1} \cap V_{P_2}, E_{P_1} \cap E_{P_2}, att_\cap, l_\cap)$,
$att_\cap(e) = att_{P_i}(e)$ and 
$l_\cap(e) = l_{p_i}(e)$ for all $e \in E_{P_1} \cap E_{P_2}$.
As $att_{P_1}$ and $att_{P_2}$ as well as $l_{P_1}$ and $l_{P_2}$ are equal on the intersection $E_{P_1} \cap E_{P_2}$, $att_\cap$ and $l_\cap$ are proper mappings.

For a set $S$ of sub-$\Sigma$-hypergraphs of a $\Sigma$-hypergraph $B$ we have
$\union(S)\colon \UNION(S) \to B$ with $\UNION(S) = (\bigcup\limits_{P \in S} V_P, \bigcup\limits_{P \in S} E_P, att_\cup, l_\cup), att_\cup(e) = att_p(e)$ and $l_\cup(e) = l_P(e)$ for all $e \in E_P, P \in S$.
As $att_{P_1}$ and  $att_{P_2}$ as well as  $l_{P_1}$ and  $l_{P_2}$ are equal on the intersection of $P_1$ and $P_2$, $att_\cup$ and $l_\cup$ are well-defined.

As a further example, we consider the category $\cat{Pos}$ of partially ordered sets (posets for short).
A poset (which can also be seen as simple acyclic transitive directed graph)
is a pair $(A,R)$ consisting of a set $A$ and a binary relation $R \subseteq A \times A$
subject to the conditions:
\begin{itemize}
\item
reflexivity, i.e., $(a,a) \in R$ for all $a \in A$,
\item
anti-symmetry, i.e., $(a,b),(b,a) \in R$ implies $a=b$ for all $a,b \in A$, and
\item
transitivity, i.e., $(a,b),(b,c) \in R$ implies $(a,c) \in R$ for all $a,b,c \in A$.
\end{itemize}
A morphism $f\colon (A,R) \to (A', R')$ is given by an order-preserving mapping $f\colon A \to A'$ meaning that 
$(f(a),f(b)) \in R'$ for all $(a,b) \in R$.
Composition and identity are the same as in $\cat{Sets}$.
A morphism $f$ is a monomorphism if and only if the underlying mapping is injective.
If $f\colon (A,R) \to (A',R')$ is a monomorphism, then the induced subobject of $(A',R')$ is represented by the poset $(f(A),f(R))$ with $f(R) = \{ (f(a),f(b)) \mid (a,b) \in R \}$.
Conversely, a poset $(A,R)$ is a subposet of the poset $(A',R')$ if $A \subseteq A'$ and $R \subseteq R'$,
denoted by $(A,R) \subseteq (A',R')$.
Then the inclusion $A \subseteq A'$ is a monomorphism such that $(A,R)$ represents a subobject of $(A',R')$.

The \emph{empty poset} $(\emptyset, \varepsilon)$, where $\varepsilon$ is the empty relation, is obviously an initial object in $\cat{Pos}$
such that the inclusion $\emptyset \subseteq A'$ provides the empty subobject $\emptymor{(A',R')}\colon (\emptyset, \varepsilon) \to (A',R')$ of each poset~$(A',R')$.

Given two subposets $(A_1,R_1),(A_2,R_2) \subseteq (A',R')$,
the intersection $(A_1,R_{1}) \cap (A_2,R_{2}) = (A_1 \cap A_2,\linebreak R_1 \cap R_2)$ is obviously a subposet of $(A',R')$ and -- together with the inclusions $(A_1 \cap A_2, R_1 \cap R_2) \subseteq (A_i,R_i)$ for $i=1,2$ -- the pullback of $(A_i,R_i) \subseteq (A',R')$ for $i=1,2$.

Given a finite set $S$ of subposets of $(A',R')$.
The union $(\bigcup\limits_{(A,R) \in S} A, \trans(\bigcup\limits_{(A,R) \in S} R))$ where
$\trans(\overline{R})$ is the transitive closure for $\overline{R} \subseteq R'$ is obviously the smallest subposet of $(A',R')$ that includes all $(A,R) \in S$ and, therefore, all pairwise intersections, too.
Altogether, this shows that $\cat{Pos}$ is an \eiu-category.
\end{example}

It may be noted that one encounters well-known categorical concepts in
the literature that are closely related to \eiu-categories (see, e.g.,~\cite{Lack-Sobocinski:05}).
\begin{enumerate}
\item
Strict initial objects are relevant for relating adhesive and extensive categories.
An initial object is \emph{strict} if every morphism into it is an
isomorphism. This implies that the initial morphisms are monomorphisms
such that the initial morphism from a strict initial object into some
object $B$ represents a subobject of $B$. The converse does not hold as the
category of pointed sets shows (cf. Example~3.8 in~\cite{Lack-Sobocinski:05}).

\item
In adhesive categories, the binary union of subobjects can be defined by
the pushout of the intersection (see Theorem~5.1 in~\cite{Lack-Sobocinski:05}).
Repeating the construction, one obtains
an iterated union of a finite set of subobjects.
It is open whether this iterated union coincides with our union.
If this is the case, then adhesive categories with empty subobjects would be \eiu-categories.
The converse does not hold as the category $\cat{Pos}$ is an \eiu-category, but it is not
adhesive (cf. Example~3.4 in~\cite{Lack-Sobocinski:05}).
\end{enumerate}

\section{Reaction Systems over \eiu-categories}
\label{sec:reaction-systems-cats}
In this section, we introduce the notion of reaction systems over an \eiu-category.
This can be done in a straightforward way
by replacing every occurrence of
“(sub)set/(sub)graph” in the definition of set/graph-based reaction
systems by “(sub)object” with one exception: the enabledness with
respect to the inhibitor.
The graph-based inhibitor (consisting of sets of
vertices and edges) has not a direct counterpart as categorical objects
do not provide explicit internal information like vertices and edges of graphs.
Therefore, we replace it by a subobject $i \colon I \to B$ of the background
like reactant and product accompanied by a subobject $i_0 \colon I_0 \to I$.
This allows to require that the intersection of $i$ and a current state is
included in $i_0$ so that the "complement" of $i$ and $i_0$ is forbidden.

\subsection{Reaction Systems over $\cat{C}$}
Let $\cat{C}$ be an \eiu-category.
Then we can define reaction systems over $\cat{C}$ in a way analogous to set-based and graph-based reaction systems.

\begin{definition}
\label{def:rs-over-c}
\rm
  \begin{enumerate}
  \item
    Let $B$ be a finite object in $\cat{C}$.
   A  \emph{reaction} over $B$ is a triple $a= \reaction B r R i I {i_0} {I_0} p P$
   where $r$ and $p$ are non-empty subobjects of $B$, $i$ is a subobject of $B$ and $i_0$ is a subobject of $I$.
   The subobject $r$ is called \emph{reactant}, the pair $(i,i_0)$ is called  \emph{inhibitor}, and $p$ is called \emph{product}. $r$, $(i,i_0)$ and $p$ may also be denoted by $r_a$, $(i_a,(i_0)_a) $ and $p_a$, respectively.

\item
A \emph{state} is a subobject of $B$.

  \item
A reaction  $a= \reaction B r R i I {i_0} {I_0} p P$ is \emph{enabled} on a state $t\colon T \to B$, denoted by $\en_a(t)$, if
$r \subseteq t$ and $t \cap i \subseteq i \circ i_0$, i.e.,
there is a monomorphism $s\colon R \rightarrow T$ with $r=t \circ s$ and, for the intersection $(T \cap I,i',t')$ of $t$ and $i$, there is a monomorphism $s'\colon T \cap I \rightarrow I_0$ with $t \cap i = i \circ i_0 \circ s'$.
   \[
    \begin{tikzcd}
    & & & T \cap I \arrow[r, "s'", dashed] \arrow[d, "i'"] \arrow[dl, "t'" above] \arrow[ddl, "t\cap i" left] \arrow[dl, phantom, "~=" below] \arrow[ddl, phantom, "=~~" below right] & I_0 \arrow[dl, "i_0"] \arrow[dl, "\text{\normalsize =}" above left] \\
      & R \arrow[dr, "r" below left] \arrow[dr, phantom, "=" above right] \arrow[r, "s", dashed] & T \arrow[d, "t"] & I \arrow[dl, "i"] \\ 
      & & B &
    \end{tikzcd}
    \]

\item
The \emph{result} of a reaction $a$ on a state $t$ is $\res_a(t)=p_a$ for $\en_a(t)$ and $\res_a(t) = \emptymor{B}$ otherwise.

\item
  Given a state $t\colon T \rightarrow B$, the \emph{result} of a set of reactions $A$ on $t$ is $\res_A(t)=\union(\{ \res_a(t) \mid a \in A \})$.
  
\item
A \emph{reaction system over} $\cat{C}$ is a pair $\mathcal{A} = (B, A)$
    consisting of some finite object $B$, called \emph{background}, and a finite set $A$ of reactions over $B$.

 \item
 Given a state $t\colon T \rightarrow B$, the \emph{result} of ${\cal A}$ on $t$ is the result of $A$ on $t$.
 It is denoted by $\res_{\cal A}(t)$.
  \end{enumerate}
\end{definition}

\begin{remark}
\label{remark:properties-rs}
Some basic properties of enabledness and results which are known for set- and graph-based reaction systems carry over to reaction systems over a category.
\begin{enumerate}
\item
A current state vanishes completely.
But it or some subobject of it may be reproduced by the products of enabled reactions.

\item
$\res_{\mathcal{A}}(t)$ is uniquely defined for every state $t$ so that
$\res_{\mathcal{A}}(t)$ is a function on the set of states of $B$.

\item
All reactions contribute to $\res_{\mathcal{A}}(t)$ in a maximally parallel and cumulative way.
There is never any conflict.

\item
As the addition of the empty subobject to a union of subobjects does not change the union,
$\res_A(t) = \res_{\{ a \in A \mid \en_a(t)\}}(t)$
holds for all states $t$.

\item
As the intersection of a subobject and the empty subobject is empty, a reaction with an empty inhibitor, i.e., $a = (r, (\emptymor{B}, \id_{\INIT}), p)$
is enabled on a state $t$ if $r \subseteq t$.
The empty inhibitor has no effect.
Therefore, the reaction is called {\em uninhibited}.

\end{enumerate}

\end{remark}

Set-based and graph-based reaction systems can be transformed into reaction systems over the categories of sets and graphs, respectively.
The transformations preserve the semantics so that set-based and graph-based reaction systems fit fully into the categorical framework.
In the following we discuss two reaction systems over $\sigmacat{Hypergraphs}$.

\subsection{Two Reaction Systems over $\sigmacat{Hypergraphs}$}
\label{sect-reaction-systems-hypergraph}
\newcommand{\rshypergraphcoloringsetflag}[2]{
\begin{tikzpicture}[baseline=-2pt]
\node[state] (1) at (0,0) {$#1$};
\node[fill=white,draw=black,rectangle] (flag) at (0.9,0) {$\mathit{#2}$};
\path (1) edge node[midway, above] {\scriptsize $1$} (flag);
\end{tikzpicture}
}

\newcommand{\rshypergraphcoloringexamplegraph}{
\foreach \pos/\x/\lab in {{(0:1.17cm)/3/3},{(72:1.17cm)/2/2},{(144:1.17cm)/1/1},{(216:1.17cm)/5/5},{(288:1.17cm)/4/4}}
{
  \node[state] (\x) at \pos {$\lab$};
}
\foreach \pos/\x/\lab in {
{(80:0.5cm)/b1/\bullet},
{(320:0.5cm)/b2/\bullet},
{(200:0.5cm)/b3/\bullet}}
{
  \node[fill=white, inner sep=0pt] (\x) at \pos {$\lab$};
}

\path (1) edge node[midway, above] {\scriptsize $1$} (b1);
\path (2) edge node[midway, right] {\scriptsize $2$} (b1);
\path (3) edge node[midway, above] {\scriptsize $3$} (b1);

\path (1) edge node[midway, above] {\scriptsize $1$} (b2);
\path (3) edge node[midway, below] {\scriptsize $2$} (b2);
\path (4) edge node[midway, right] {\scriptsize $3$} (b2);

\path (1) edge node[midway, below] {\scriptsize $1$} (b3);
\path (5) edge node[midway, right] {\scriptsize $2$} (b3);
\path (4) edge node[midway, above] {\scriptsize $3$} (b3);
}

\newcommand{\rshyperedgecoveringexamplegraph}{
\foreach \pos/\x/\lab in {{(0:1.17cm)/3/3},{(72:1.17cm)/2/2},{(144:1.17cm)/1/1},{(216:1.17cm)/5/5},{(288:1.17cm)/4/4}}
{
  \node[state] (\x) at \pos {$\lab$};
}
\foreach \pos/\x/\lab in {
{(80:0.5cm)/b1/\bullet},
{(320:0.5cm)/b2/\bullet},
{(200:0.5cm)/b3/\bullet}}
{
  \node[fill=white,,draw=black,rectangle, inner sep=2pt] (\x) at \pos {$*$};
}
}

As a first example, we model a vertex-coverability test by a family of reaction systems over the category $\sigmacat{Hypergraphs}$.

Let $H = (V,E,att, l)$ be a $\Sigma$-hypergraph with $l(e)=*$ for some label $* \in \Sigma$ for all $e \in E$
(this means that all hyperedges are equally labeled and, hence, can be considered as unlabeled).
Then $X \subseteq V$ is a \emph{vertex cover} of $H$ if each hyperedge has some attachment vertex in $X$.
$H$ is $k$-\emph{vertex-coverable} for some $k \in \Nat$ if there is a hyperedge vertex cover of $H$ with $k$ elements.

The $k$-vertex-coverability test employs the reaction system $\mathcal{A}_{m,n} = (B_{m,n}, \allowbreak A_{m,n})$ for some $m,n \in \Nat$ with $m \le n$ defined as follows.
Let $\binomb{n}{m}$ be the set of all strings over $[n]$ of lengths up to $m$.
Then the \emph{complete hypergraph with twins}
is defined by $CH^{(2)}_{m,n} = ([n], \binomb{n}{m} \times \{ *, + \}, attach, lab)$ with $attach(u,*) = attach(u,+)  = u$ and $lab(u,*) = *$ and $lab(u,+) = +$ for all $u \in \binomb{n}{m}$.
The two parallel hyperedges $(u,*)$ and $(u,+)$ for $u \in \binomb{n}{m}$ are called \emph{twins}.
The background hypergraph $B_{m,n}$ is $CH^{(2)}_{m,n}$ extended by a $*$-flag (type-1 hyperedge) at each vertex.
The set of reactions $A_{m,n}$ contains the following elements,
where, due to the one-to-one correspondence of categorial subobjects of a $\Sigma$-hypergraph and sub-$\Sigma$-hypergraphs, the subobjects are represented by the domain objects of the inclusion morphisms.
The symbol ``$-$'' is a shortcut for the inhibitor $(\emptymor{B_{m,n}}, \id_{\MPT})$.
\begin{enumerate}
\item $(\begin{tikzpicture}[->,>=stealth',shorten >=1pt,auto,node distance=1.8cm,baseline = -3pt]
  
  \node[state] (j)                    {$j$};
\end{tikzpicture}
,-, \begin{tikzpicture}[->,>=stealth',shorten >=1pt,auto,node distance=1.8cm,baseline = -3pt]
  
  \node[state] (j)                    {$j$};
\end{tikzpicture}
)$ for all $j \in [n]$.
\item $(e^\bullet,-, e^\bullet)$ for all $e \in \binomb{n}{m} \times \{ *, + \}$ where $e^\bullet$ is the sub-$\Sigma$-hypergraph of $B_{m,n}$ induced by $e$, i.e., $e^\bullet = (\{v_1,\ldots, v_l \}, \{e\}, attach|_{\{e\}}, lab|_{\{e\}})$ with $attach(e) = v_1\cdots v_l$, $v_j \in [n]$ for $j=1,\ldots, l$.
\item $(\rshypergraphcoloringsetflag{j}{*},-, \rshypergraphcoloringsetflag{j}{*})$ for all $j \in [n]$.
\item $((u,*)^\bullet \cup v^\bullet,-, (u,+)^\bullet)$ for all $u \in \binomb{n}{m}$ and $v \in V$ occurring in $u$
  where $v^\bullet$ is the sub-$\Sigma$-hypergraph of $B_{m,n}$ with the vertex $v$ and a $*$-flag at $v$.
\end{enumerate}
The first three types of reactions applied to a state make sure that the state is sustained.
The only changing reactions are of the fourth type.
They add a $+$-labeled twin hyperedge whenever some attachment vertex of a $*$-labeled hyperedge has a $*$-flag.
In the drawings, a circle represents a vertex and a box a flag. The label is inside the box, and a line from a box to a circle represents the attachment.

The modeling is continued in Section~\ref{sec:reaction-systems-hypergraph-interactive-process}.

\bigskip

The second example is less interesting from a computational
point of view, but serves to illustrate how non-trivial inhibitors work.
Consider $CH^{(2)}_{m,n}$ as background and the following reactions:
$a(e)=(e^\bullet , \{v_1,\ldots,v_l\} \subset \hat{e}^\bullet , e^\bullet)$
for all $e \in \binomb{n}{m}\times \{*\}$ with $attach(e)=v_1\cdots v_l$ where
$\hat{e}$ is the twin of $e, e^\bullet$ and $\hat{e}^\bullet$ are
defined as in Point~1, and $\{v_1,\ldots,v_l\}$ represents the discrete hypergraph with
the attachment vertices of $e$ as vertices.
A reaction $a(e)$ is enabled on
some state $H$ if $e \in E_H$ and $\hat{e} \notin E_H$.
In other words,
the application of all reactions sustains all $*$-hyperedges of $H$ that are
not accompanied by their twins.

\subsection{Interactive Processes}

The definition of reaction systems over a category is chosen in such a way that the semantic notion of interactive processes can be carried over directly from the set-based and graph-based cases.

\begin{definition}
\begin{enumerate}
\item
Let $\mathcal{A}  = (B, A)$ be a reaction system over $\cat{C}$.
An \emph{interactive process} $\pi =(\gamma,\delta)$ on~$\mathcal{A}$ consists of two sequences of subobjects of $B$ $\gamma = c_0,\dots,c_n$ and $\delta = d_0,\dots,d_n$ for some $n \geq 1$ such that $d_i = \res_{\cal A}(c_{i-1} \cup d_{i-1})$ for $i=1,\dots,n$.
The sequence $\gamma$ is called \emph{context sequence}, the sequence $\delta$ is called \emph{result sequence}
where $d_0$ is called \emph{start}, and the sequence $\tau = t_0, \dots, t_n$ with $t_i = c_i \cup d_i$ for $i=0,\ldots,n$ is called \emph{state sequence}.

\item
$\pi$ is called \emph{context-independent} if $c_i \subseteq d_i$ for $i=0,\ldots,n$.

\end{enumerate}
\end{definition}

\begin{remark}
\label{remark:properties-interactive-processes}
Consider a context-independent process
$\pi=(c_0,\ldots,c_n, d_0,\ldots,\allowbreak d_n)$.
\begin{enumerate}
\item
Using point~\ref{item:union-7} of Properties~\ref{properties:empty-objects-intersection-union} in the
previous section, $c_i \subseteq d_i$ for $i=0,\ldots,n$ implies $t_i=c_i \cup d_i=d_i$
meaning that the result sequence and state sequence coincide and that
the state sequence describes the whole process determined by its initial
state $t_0 = d_0$.
Therefore, whenever context-independent processes are
considered, one can focus on their state sequences.

\item
Let $\tau = t_0,\ldots,t_n$ for some $n\ge 1$ be a state sequence.
Then $\tau$ is either \emph{repetition-free}, i.e., $t_i \ne t_j$ for all $i,j$ with $0 \le i < j \le n$,
or there is a smallest pair $t_{i_0},t_{j_0}$ with $0 \le i_0 < j_0 \le n$ and $t_{i_0} = t_{j_0}$ such that
$\tau = t_0,\ldots,t_{i_0}, (t_{i_0 + 1},\ldots, t_{j_0})^m t_k,\ldots, t_n$ for some $m \in \Nat$ where $k=i_0 + 1 + m(j_0 - i_0) + 1$
and $t_k,\ldots, t_n$ is an initial section of $t_{i_0 + 1},\ldots, t_{j_0}$.
According to the choice of $i_0$ and $j_0$, the section $t_0,\ldots, t_{j_0 -1}$ is repetition-free.

\item
Using the pigeonhole principle, the pair $i_0,j_0$ exists if $n-1$ is greater than the number of states.
Therefore, every state sequence runs into a unique cycle eventually.

\end{enumerate}
\end{remark}


\subsection{An Interactive Process for $\sigmacat{Hypergraphs}$}
\label{sec:reaction-systems-hypergraph-interactive-process}
Let $H \subseteq CH^{(2)}_{m,n}$ be a sub-$\Sigma$-hypergraph with $*$-labeled hyperedges only.
Let $i_1,\ldots,i_k$ be a combination of $k$ elements of $[n]$ for some $k \in \Nat$.
Then one can consider the interactive process $\pi(H,i_1 \cdots i_k) = (\gamma(H,i_1\cdots i_k), \delta(H, i_1\cdots i_k))$ with
$\gamma(H, i_1\cdots i_k) = \rshypergraphcoloringsetflag{i_1}{*}, \ldots,\allowbreak \rshypergraphcoloringsetflag{i_k}{*}, \MPT$ and $H$ as start.
Then $\{i_1,\ldots,i_k \}$ is a $k$-vertex-cover of $H$ if and only if each hyperedge of $H$ has a twin in the final result.
Consequently, to test whether $H$ is $k$-vertex-coverable, one may run the interactive process $\pi(H, i_1\cdots i_k)$ for all combinations of $k$ elements of~$[n]$.

\begin{example}
Let $\gamma(B_{3,5},2,4) =\rshypergraphcoloringsetflag{2}{*}, \rshypergraphcoloringsetflag{4}{*} ,\MPT$.
The result sequence is
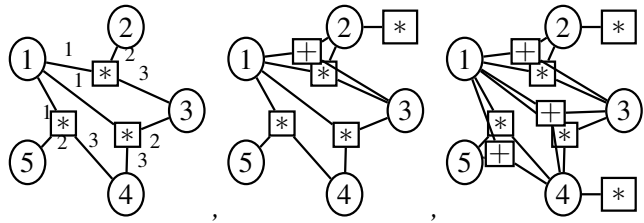
\begin{center}
\begin{tikzpicture}
  \rshyperedgecoveringexamplegraph

\path (1) edge node[midway, above] {\scriptsize $1$} (b1);
\path (2) edge node[midway, right] {\scriptsize $2$} (b1);
\path (3) edge node[midway, above] {\scriptsize $3$} (b1);

\path (1) edge node[midway, above] {\scriptsize $1$} (b2);
\path (3) edge node[midway, below] {\scriptsize $2$} (b2);
\path (4) edge node[midway, right] {\scriptsize $3$} (b2);

\path (1) edge node[midway, below] {\scriptsize $1$} (b3);
\path (5) edge node[midway, right] {\scriptsize $2$} (b3);
\path (4) edge node[midway, above] {\scriptsize $3$} (b3);

\end{tikzpicture}
,
\begin{tikzpicture}
  \rshyperedgecoveringexamplegraph

\path (1) edge  (b1);
\path (2) edge  (b1);
\path (3) edge  (b1);

\path (1) edge  (b2);
\path (3) edge  (b2);
\path (4) edge  (b2);

\path (1) edge  (b3);
\path (5) edge  (b3);
\path (4) edge  (b3);

\node[fill=white,draw=black,rectangle] (flag1)[right = 0.25 of 2] {$*$};
\path (2) edge  (flag1);

\foreach \pos/\x in {
{(100:0.8cm)/b1}}
{
  \node[fill=white,,draw=black,rectangle, inner sep=1pt] (\x) at \pos {$+$};
}

\path (1) edge (b1);
\path (2) edge (b1);
\path (3) edge (b1);

\end{tikzpicture}
,
\begin{tikzpicture}
  \rshyperedgecoveringexamplegraph

\path (1) edge  (b1);
\path (2) edge  (b1);
\path (3) edge  (b1);

\path (1) edge  (b2);
\path (3) edge  (b2);
\path (4) edge  (b2);

\path (1) edge  (b3);
\path (5) edge  (b3);
\path (4) edge  (b3);

\node[fill=white,draw=black,rectangle] (flag1)[right = 0.25 of 2] {$*$};
\path (2) edge  (flag1);

\node[fill=white,draw=black,rectangle] (flag1)[right = 0.25 of 4] {$*$};
\path (4) edge  (flag1);

\foreach \pos/\x in {
{(100:0.8cm)/b1},
{(350:0.2cm)/b2},
{(230:0.75cm)/b3}}
{
  \node[fill=white,,draw=black,rectangle, inner sep=1pt] (\x) at \pos {$+$};
}

\path (1) edge (b1);
\path (2) edge (b1);
\path (3) edge (b1);

\path (1) edge (b2);
\path (3) edge (b2);
\path (4) edge (b2);

\path (1) edge (b3);
\path (4) edge (b3);
\path (5) edge (b3);

\end{tikzpicture}
.
\end{center}
The lines connecting a box with vertex circle provide the attachment where the numbering establishes its order.
In the second and third hypergraph the numbering is omitted to clarify the drawing.
$c_0$ enables the reaction $((123,*)^\bullet \cup 2^\bullet,-, (123,+)^\bullet)$ and
$c_1$ enables the reaction $((134,*)^\bullet \cup 4^\bullet,-, (134,+)^\bullet)$ as well as the reaction $((145,*)^\bullet \cup 4^\bullet,-, (145,+)^\bullet)$.

Note that it is also possible to choose both in parallel, e.g., choose $c'_0$ to be $c_0 \cup c_1 = \rshypergraphcoloringsetflag{2}{*}~\rshypergraphcoloringsetflag{4}{*}$
and $c'_1 = \MPT$ meaning that the test can be done in one step.
\end{example}

\section{Diagram Categories are \eiu-categories}
\label{sec:further-look-at-the-categorial-framework}
Many categories follow a common building principle, called diagram categories, providing a reservoir of potential example categories over which reaction systems are defined because certain diagram categories turn out to be
\eiu-categories if the underlying category is an \eiu-category.


Let $\scheme = (C,A,s\colon A \to C, t\colon A \to C)$ be a directed graph (without labeling), called \emph{scheme}, where the vertices are also called \emph{components} and the edges \emph{arrows}.
Then $\scheme$ induces the \emph{diagram category} $\diagramcat{C}{\scheme}$ over $\cat{C}$.
Its objects are graph morphisms $\delta\colon \scheme \to gr(\cat{C})$, where the domain is the scheme $\scheme$ and the codomain is the underlying graph of the category $\cat{C}$, i.e.,
$gr(\cat{C}) = (\Ob_{\cat{C}}, \sum\limits_{X,Y \in \Ob_{\cat{C}}} \Mor_{\cat{C}}(X,Y), \hat{s}, \hat{t})$ with objects of $\cat{C}$ as vertices and the disjoint union of all sets of morphisms as set of edges,
and $\hat{s}(f\colon X \to Y) = X$ and $\hat{t}(f\colon X \to Y) = Y$
for all $f \in \Mor_{\cat{C}}(X,Y)$ and all $X,Y \in \Ob_{\cat{C}}$.
The objects of $\diagramcat{C}{\scheme}$ are called \emph{diagrams}.
Given two diagrams $\delta,\delta'\colon \scheme \to gr(\cat{C})$,
a morphism $g \colon \delta \to \delta'$ is given by a family of $\cat{C}$-morphisms
$\{g_c \colon \delta_V(c) \to \delta'_V(c)\}_{c\in C}$ such that
$g_{t(a)} \circ \delta_E(a) = \delta'_E(a) \circ g_{s(a)}$ for all $a \in A$.
This means that the following diagram commutes: 
\[
\begin{tikzcd}
\delta_V(s(a)) \arrow[d, "g_{s(a)}"] \arrow[r, "\delta_E(a)"] \arrow[dr, phantom, "=" description] & \delta_V(t(a)) \arrow[d, "g_{t(a)}"]\\
\delta'_V(s(a)) \arrow[r, "\delta'_E(a)"] & \delta'_V(t(a))
\end{tikzcd}
\]

The composition and the identities are defined componentwise in the category $\cat{C}$.
The components of $\scheme$ are placeholders for objects, the arrows for morphisms.
To avoid an extra handling of labeling and typing functions or such, we also allow fixed components meaning that such a component is instantiated by some fixed object in each diagram
and each morphism in a fixed component is always the identity.

Schemes may be drawn in the usual way:
Bullets represent components connected by arrows from source bullet to target bullet each.
In the case of a fixed component, the bullet is replaced by the associated fixed object.

Often used categories turn out to be diagram categories:
\begin{enumerate}
\item
The product category $\cat{Sets} \times \cat{Sets} = \diagramcat{Sets}{\bullet \bullet}$ of ordered pairs of sets.

\item
The category $\sigmacat{Sets}= \diagramcat{Sets}{\bullet \rightarrow \Sigma}$ of $\Sigma$-labeled sets for some alphabet $\Sigma$.

\item
The category $\cat{Maps} = \diagramcat{Sets}{\bullet \rightarrow \bullet}$ of mappings.

\item
The category $\cat{Graphs} = \diagramcat{Sets}{\bullet \xrightrightarrows[]{} \bullet}$ of directed (unlabeled) graphs.

\item
The category
$\sigmacat{Graphs} = \diagramcat{Sets}{\Sigma \leftarrow \bullet \xrightrightarrows{} \bullet}$
of $\Sigma$-graphs for some alphabet $\Sigma$.

\item The category $(\Sigma_V, \Sigma_E)$-$\cat{Graphs} = \diagramcat{Sets}{\Sigma_E \leftarrow \bullet \xrightrightarrows[]{} \bullet \rightarrow \Sigma_V}$
of directed vertex- and edge-labeled graphs.

\item The category $\cat{BipartiteGraphs} =  \diagramcat{Sets}{\begin{tikzpicture}[>= stealth',every state/.style = {circle,draw,minimum size = 1cm, inner sep = 0cm,fill}, every edge/.style = {draw,inner sep = 1.5pt, line width=0.1mm},on grid,inner sep = 0.03cm]

  \node [circle,fill=black] (1) at (0,0) {};
  \node [circle,fill=black] (2) at (0.4,0) {};
  \node [circle,fill=black] (3) at (0,0.3) {};
  \node [circle,fill=black] (4) at (0.4,0.3) {};

  \path (1) edge[arrows={-to}] (3);
  \path (1) edge[arrows={-to}] (4);
  \path (2) edge[arrows={-to}] (4);
  \path (2) edge[arrows={-to}] (3);
\end{tikzpicture}}$ of bipartite directed graphs.
Let $G = (V_1,V_2,E_1,E_2,\allowbreak s_1\colon E_1 \to V_1,s_2\colon E_2 \to V_2,t_1\colon E_1 \to V_2,t_2\colon E_2 \to V_1)$ be an object.
There are two sets of vertices and two sets of edges.
Edges have sources in $V_1$ and targets in $V_2$ or the other way round.

\item The category $3$-$\cat{Hypergraphs} = \diagramcat{Sets}{\bullet \threerightarrows \bullet}$  of hypergraphs with hyperedges of type 3.
Let the three arrows be $l,r,t$ respectively, and let $H = (V,E,l_H, r_H, t_H)$ be an object.
Then each $e \in E$ is attached to a ``left'', a ``right'', and a ``top'' vertex so that $e$ can be seen as a triangle.

\item 
The category $4$-$\cat{Hypergraphs} = \diagramcat{Sets}{\bullet \fourrightarrows \bullet}$  of hypergraphs with hyperedges of type 4.
Let the four arrows be $\mathit{north},\mathit{east},\mathit{south},\mathit{west}$, then the hyperedges are of type 4 and can be seen as ``cells'' with ``tentacles'' to the respective directions.

\item
  An interesting example where the underlying category is not $\cat{Sets}$ is the category
  $\typedgraphsdiagramcat{}$ of $\mathit{TG}$-\emph{typed graphs}
  for some \emph{type graph} $\mathit{TG}$.
  They are often used in the area of graph transformation as a well-working generalization of labeled graphs.
  A $\mathit{TG}$-typed graph is represented by a pair $(G,t)$, where $G$ is a directed (unlabeled) graph and $t\colon G \to \mathit{TG}$ is a graph morphism specifying the structure of $G$.
  A $\mathit{TG}$-type-graph morphism $f\colon (G_1, t_1) \to (G_2, t_2)$ is a graph morphism \mbox{$f_G\colon G_1 \to G_2$} such that $t_2 \circ f_G = t_1$.

  Indeed, $\sigmacat{Graphs}$ is in a one-to-one correspondence to $\typedgraphsdiagramcat{(\Sigma)}$ where $\mathit{TG}(\Sigma)$ has a single vertex and, for each $x \in \Sigma$, an $x$-labeled loop at the vertex.
  Similarly, $(\Sigma_V, \Sigma_E)$-$\cat{Graphs}$ is in a one-to-one correspondence to $\typedgraphsdiagramcat{(\Sigma_V,\Sigma_E)}$
  where $\mathit{TG}(\Sigma_V,\Sigma_E) = (\Sigma_V, \Sigma_V \times \Sigma_E \times \Sigma_V, pr_1, pr_3)$ with the first and third projections $pr_1$ and $pr_3$ as source and target mappings respectively.

\end{enumerate}

Concerning diagram categories, it may be noted that categories of the form
$\diagramcat{C}{\bullet \to X}$ for some fixed object $X$ are also called \emph{slice categories}.
Two of our examples,
$\sigmacat{Sets} = \diagramcat{Sets}{\bullet \rightarrow \Sigma}$ and
$\cat{TypedGraphs} = \diagramcat{Graphs}{\bullet \rightarrow \mathit{TG}}$ 
are slice categories.

The main result of this section is that 
diagram categories are \eiu-categories
if the underlying category is an \eiu-category and the fixed components of the considered schemes have no out-going arrows.

\begin{theorem}\label{thm:diagramcats-of-eiu-cats-are-eiu-cats}
Let $\cat{C}$ be an \eiu-category
and $\scheme$ be a scheme where no fixed component is a source of an arrow.
Then $\diagramcat{C}{\scheme}$ is an \eiu-category.
\end{theorem}

\begin{proof}
It is known that limits and colimits in diagram categories without fixed
components can be constructed componentwise by limits and colimits of
the underlying category. It is also known that limits and colimits in
slice categories (with a scheme of the form $\bullet \to \Sigma$) can be
constructed by the limits and colimits of the free component in the
underlying category. The statement can be proved for diagram categories
with fixed components in the same way by combining the arguments for the
two known cases.
\end{proof}

\begin{remark}
The proof of the theorem is not only analogous to the proof for
diagram categories without fixed components and slice categories, but it
also may be that a diagram category with fixed components is isomorphic
to a slice category with an underlying diagram category witout fixed
components so that the theorem follows directly from the known results.
\end{remark}

If one allows to replace a bullet in a scheme $\scheme$ by a $*$ and uses it in $\diagramcat{Sets}{\scheme}$
in such a way that the $*$ is not replaced by a set $X$, but by the set of strings $X^*$ over $X$, then even the category of $\sigmacat{Hypergraphs}$ can be obtained as a diagram category:
$\sigmacat{Hypergraphs} = \diagramcat{Sets}{\Sigma \leftarrow \bullet \rightarrow{} *}$.
We know already that $\sigmacat{Hypergraphs}$ is an \eiu-category.
But it is open whether Theorem~\ref{thm:diagramcats-of-eiu-cats-are-eiu-cats} holds using this kind of schemes.

It is not difficult to see that all our explicit
examples listed above and many like these meet the assumptions of the
theorem.

\section{Towards a Category of Reaction Systems over a Category}
\label{sec:category-reaction-systems-over-cat}
So far, everything we have discussed concerns reactions systems
over categories. But there are more ways to bring reaction systems and
category theory together.
Whenever one has
a class of entities, one may try to use them as objects of a category by
choosing suitable morphisms. Therefore, one may ask how reaction systems
over a category may be provided with a meaningful notion of morphisms.

In this section, we show that, given a reaction system $\mathcal{A}=(B,A)$ over $\cat{C}$,
a monomorphism $f\colon B \to B’$ induces a reaction system $f(\mathcal{A})$
by composing all the components of reactions with $f$.
This observation motivates us to consider such a
morphism as morphism from $\mathcal{A}$ to $\mathcal{A'} = (B',A')$
provided that $f(A)\subseteq A'$.




\begin{theorem}
\label{thm:cat-rs-mono}
Given a reaction system $\mathcal{A} = (B,A)$ and
a monomorphism $f\colon B \to B'$.
Then $f$ induces a reaction system
$f(\mathcal{A}) = (B', f(A))$
where $f(A) = \{ f(a) \mid a \in A \}$ and $f(a)= \reaction {B'} {f\circ r} R {f\circ i} I {i_0} {I_0} {f\circ p} P$ for $a= \reaction B r R i I {i_0} {I_0} p P$.

$f(\mathcal{A})$ has the following properties.
\begin{enumerate}
\item
$en_a(t)$ on a state $t\colon T \to B$ if and only if $en_{f(a)}(f \circ t)$,

\item
$f \circ res_a(t) = res_{f(a)}(f\circ t)$,

\item
$f \circ res_{\mathcal{A}}(t) = res_{f(\mathcal{A})}(f \circ t)$.

\end{enumerate}
\end{theorem}

The following diagram shows $a$ and $f(a)$.
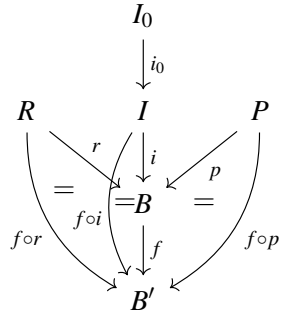
\[
\begin{tikzcd}
  & I_0 \arrow[d, "i_0"] & \\
R \arrow[dr, "r" ] \arrow[ddr, bend right, "f \circ r" below left] \arrow[ddr, phantom, "=" above left] & I \arrow[d, "i" ] \arrow[dd, bend right, "f \circ i" below left] \arrow[dd, phantom, "=" left] & P \arrow[dl, "p" ] \arrow[ddl, bend left, "f \circ p"  below right] \arrow[ddl, phantom, "=" description]\\
  & B \arrow[d, "f" ] & \\
  & B' &
\end{tikzcd}
\]

The proof uses the following lemma.
\begin{lemma}
\label{lemma:pullback-union-extended-with-monos}
\begin{enumerate}
\item
Let $p\colon P \to B, \overline{p}\colon \overline{P} \to B$ and $f\colon B \to B'$ be monomorphisms.
Let $L = (\PB, p'\colon \PB \to P, \overline{p}'\colon \PB \to \overline{P})$ be a triple of an object $\PB$ and two monomorphisms $p'$ and $\overline{p}'$.
Then $L$ is a pullback of $p$ and $\overline{p}$ if and only if $L$ is a pullback of $f \circ p$ and $f \circ \overline{p}$.

\item
Let $S$ be a set of subobjects of $B$ and $f\colon B \to B'$ be a monomorphism.
Then $f \circ \union(S) = \union( \{ f \circ p \mid p \in S  \})$.

\end{enumerate}
\end{lemma}

\begin{proof}
Point 1 follows immediately from the observation that a monomorphism is a limit and that limits compose.


2.
Using Point 1, one get $\PB(S) = \PB(\{ f \circ p \mid p \in S \})$
so that $\UNION(S) \iso \UNION(\{ f \circ p \mid p \in S  \})$ and $f \circ \union(S) = \union( \{  f \circ p \mid p \in S \})$ as subobjects.
\end{proof}

The situation of Point 1 of the lemma is depicted in the following diagram.
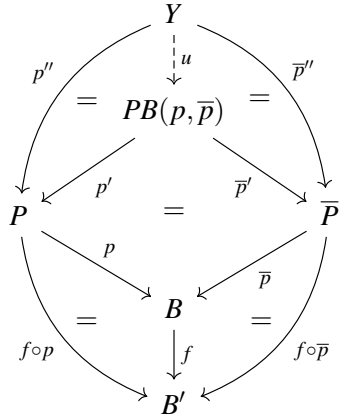
\[
\begin{tikzcd}
& Y \arrow[d, dashed, "u" ] \arrow[ddl, phantom, bend right, "~=" below right] \arrow[ddl, bend right, "p''" above left]  \arrow[ddr, bend left, "\overline{p}''" above right] \arrow[ddr, phantom, bend left, "=~" below left]  & \\
& \PB(p,\overline{p})\arrow[dr, "\overline{p}'" below left]\arrow[dl, "p'"] \arrow[dd, phantom, "="] & \\
P\arrow[dr, "p"] \arrow[ddr, bend right, "f \circ p" below left] \arrow[ddr, phantom, bend right, "~=" above right] && \overline{P}\arrow[dl, "\overline{p}"] \arrow[ddl, bend left, "f \circ \overline{p}"  below right] \arrow[ddl, phantom, bend left, "=~" above left]\\
& B \arrow[d, "f" ] & \\
& B' &
\end{tikzcd}
\]
And the situation of Point 2 where $S=\{p_1,p_2,p_3\}$ is depicted in the following diagram.
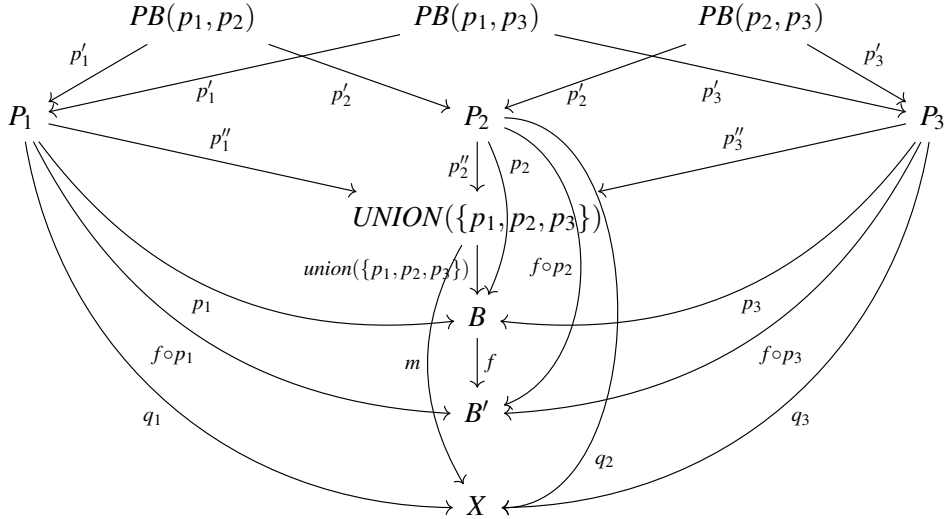
\[
\begin{tikzcd}
& \PB(p_1,p_2) \arrow[dl, "p'_1" above left] \arrow[dr, "p'_2" below left] & \PB(p_1,p_3) \arrow[dll, "p'_1" below left] \arrow[drr, "p'_3" below left] & \PB(p_2,p_3) \arrow[dl, "p'_2" below left] \arrow[dr, "p'_3" above right] &\\
P_1 \arrow[ddddrr, bend right = 40, "q_1" below left] \arrow[dddrr, bend right, "f \circ p_1" below left] \arrow[ddrr, bend right, "p_1" below left] \arrow[drr, "p_1''" above right] & & P_2 \arrow[dddd, bend left = 90, "q_2" near end] \arrow[ddd, bend left = 70, "f \circ p_2" left] \arrow[dd, bend left = 25, "p_2" near start]  \arrow[d, "p_2''" left] & & P_3 \arrow[ddddll, bend left = 40, "q_3" below right] \arrow[dddll, bend left, "f \circ p_3" below right] \arrow[ddll, bend left, "p_3" below right] \arrow[dll, "p_3''" above left]\\ %
& & \UNION(\{p_1,p_2,p_3\}) \arrow[d, "\union(\{\text{$p_1,p_2,p_3$}\})" left] \arrow[ddd, bend right, "m" left] & & \\ 
&& B \arrow[d, "f" ] && \\
&& B' && \\
&& X &&&
\end{tikzcd}
\]

\begin{proof} of Theorem~\ref{thm:cat-rs-mono}.
1.
Given a reaction $a = (r, (i,i_0), p)$ and a state $t$ in $\mathcal{A}$,
$en_a(t)$ means $r \subseteq t$ and $t \cap i \subseteq i \circ i_0$.
By definition, there are monomorphisms $s$ and $s'$ with $r=t \circ s$ and $t \cap i = i \circ i_0 \circ s'$.
This implies $f \circ r = f \circ t \circ s$ and by Point~1 of Lemma~\ref{lemma:pullback-union-extended-with-monos} also
$f \circ r \subseteq f \circ t$ and $(f \circ t) \cap (f \circ i) = f \circ (t\cap i) = f \circ i \circ i_0 \circ s'$.
This means $(f \circ t) \cap (f \circ i) \subseteq f \circ i \circ i_0$, and, therefore, $en_{f(a)}(f \circ t)$.

Conversely. $en_{f(a)}(f \circ t)$ means $f \circ r \subseteq f \circ t$ and $(f\circ t) \cap (f \circ i) \subseteq f \circ i \circ i_0$.
By definition, there are monomorphisms $s$ and $s'$ with $f \circ r = f \circ t \circ s$ and $(f\circ t) \cap (f \circ i) = f \circ i \circ i_0 \circ s'$.
By the Point~1 of~Lemma~\ref{lemma:pullback-union-extended-with-monos}
one has $f \circ (t \cap i ) = (f \circ t) \cap (f \circ i)$ so that the monomorphisms of $f$ yields $r = t \circ s$ and $t \cap i = i \circ i_0 \circ s'$.
This means $r \subseteq t$ and $t \cap i \subseteq i \circ i_0$ and, therefore, $en_a(t)$.

2.
According to Point 1, there are two cases to consider using the definition of results:
$f \circ res_a(t) = f \circ p = res_{f(a)}(f \circ t)$ provided that $a$ is enabled on $t$ and $f(a)$ on $f\circ t$;
and
$f \circ res_a(t) = f \circ \emptymor{B} = \emptymor{B'} = res_{f(a)}(f \circ t)$ otherwise.

3.
Using the definition of results of reaction systems and sets of reactions as well as Points 1 and 2
of Lemma~\ref{lemma:pullback-union-extended-with-monos}, one gets as stated:
$f \circ res_{\mathcal{A}}(t) = f \circ  res_A(t) = f \circ \union( \{ res_a(t) \mid a \in A  \} ) = \union( \{ f \circ res_a(t) \mid a \in A  \} ) = res_{f(A)}(f \circ t) = res_{f(\mathcal{A})}(f \circ t)$.
\end{proof}

Using Point 3 of the theorem and Point~\ref{item:union-8}
of Properties~\ref{properties:empty-objects-intersection-union},
one gets the following result.

\begin{corollary}
\label{corollary:rs-cat-mono}
Let
$\mathcal{A} = (B,A)$ and $\mathcal{A'} = (B',A')$
be two reaction systems over $\cat{C}$ with
$f(A) \subseteq A'$
for some monomorphism $f\colon B \to B'$.
Then
$f \circ res_{\mathcal{A}}(t) \subseteq res_{\mathcal{A'}}(f \circ t)$ for all states $t\colon T \to B$.
\end{corollary}

This motivates to define the category $\cat{RS(C)}$.

\begin{definition}
Let $\cat{C}$ be an \eiu-category.
\begin{enumerate}
\item
The category $\cat{RS(C)}$ is defined as follows.
Its objects are reactions systems over $\cat{C}$.
Given two reaction systems $\mathcal{A} = (B,A)$ and $\mathcal{A'} = (B',A')$ over $\cat{C}$,
a morphisms $f\colon \mathcal{A} \to \mathcal{A'}$ is given by monomorphisms $f\colon B \to B'$ 
provided that $f(A)\subseteq A'$.
Compositions and identities are given by the underlying morphisms.

\item
If $f \circ res_{\mathcal{A}}(t) = res_{\mathcal{A'}}(f \circ t)$ for all states $t\colon T \to B$, then $f\colon \mathcal{A} \to \mathcal{A'}$ is called \emph{strong}.
\end{enumerate}
\end{definition}

The definition of composition and identities is meaningful as, for reaction systems
$\mathcal{A} = (B,A)$, $\mathcal{A'} = (B',A')$ and $\mathcal{A''} = (B'',A'')$
and for morphisms $f\colon \mathcal{A} \to \mathcal{A'}$ and $g\colon \mathcal{A'} \to \mathcal{A''}$,
$(g \circ f )(A) = g(f(A)) \subseteq g(A') \subseteq A''$
and
$\id_B(A) = A$.

\begin{example}
Consider the two reaction systems over $\sigmacat{Hypergraphs}$
$\mathcal{A}_{m,n},  \mathcal{A}_{m',n'}$ with $m\le m'$ and $n \le n'$
as defined in Section~\ref{sect-reaction-systems-hypergraph}.

The inclusion of $B_{m,n}$ into $B_{m',n'}$ induces a morphism from
$\mathcal{A}_{m,n}$ to $\mathcal{A}_{m',n'}$ as $A_{m,n} \subseteq A_{m',n'}$.
This morphism is strong as one can see as follows.
Let $T$ be a sub-$\Sigma$-hypergraph of $B_{m,n}$
representing a state of $\mathcal{A}_{m,n}$.
Then $T$ represents also a state of $\mathcal{A}_{m',n'}$.
According to Corollary~\ref{corollary:rs-cat-mono}
we know that
$res_{\mathcal{A}_{m,n}}(T) \subseteq res_{\mathcal{A}_{m',n'}}(T)$.
Let now $(R',-,P')$ be a reaction in $\mathcal{A}_{m',n'}$
that is not in $\mathcal{A}_{m,n}$.
As $B_{m,n}$ is
complete with respect to hyperedges including flags, $R'$ and $P'$ must
contain a vertex $k'>n$.
Consequently, $R' \not\subseteq T$ such that the
reaction is not enabled and none of those can contribute to
$res_{\mathcal{A}_{m,n}}(T)$ meaning that
$res_{\mathcal{A}_{m,n}}(T) = res_{\mathcal{A}_{m',n'}}(T)$.
Summarizing, the family
$\{\mathcal{A}_{m,n}\}_{m,n \in \Nat}$
forms a two-dimensional
grid connected by strong morphisms along growing indices. This is
interesting with respect to the vertex-coverability of hypergraphs. Each
hypergraph can be transformed into a sub-$\Sigma$-hypergraph of
$B_{m,n}$ for
some $m,n$ by numbering the vertices and removing labels, multiples of
hyperedges and multiples of vertex attachments within a hyperedge in such
a way that its vertex-coverability is preserved. Then the grid
of strong morphisms makes sure that the result of the vertex-coverability
test is independent of the choice of the $B_{m,n}$ as long as the
transformation works. In this sense, the family
$\{\mathcal{A}_{m,n}\}_{m,n \in \Nat}$
models a vertex-coverability test for all hypergraphs.
\end{example}




\section{Conclusion}
\label{sec:conclusion}
In this paper, we have proposed a categorical framework for the modeling
of reaction systems. We have provided appropriate categorical notions
including finite objects, subobjects, subobject inclusions, empty subobjects,
intersections and unions of subobjects that allow the
definition of reaction systems over \eiu-categories and their
interactive-process semantics in a quite similar way to the known set-
and graph-based reactions systems. Moreover, we have shown that many
categories meet the categorical requirements so that many structures
become available on which reaction systems may be based on. This
includes, in particular, quite a variety of graphs, hypergraphs, and
other graph-like structures. But we have only done the very first steps
into a categorical approach. To shine more light on the significance of
the framework, the investigation should be continued including the
following topics.

\begin{enumerate}

\item
As pointed out at the end of Section~\ref{sec:preliminaries}, it
would be interesting to clarify the relationship between
\eiu-categories and the well-studied  adhesive categories that are
successfully applied in the area of graph transformation in various
variants (cf., e.g.,~\cite{Lack-Sobocinski:05,Ehrig-Ehrig-Prange.ea:06,Corradini-Herman-Sobocinski:08,Braatz-Ehrig-Gabriel-Golas:10,Ehrig-Ermel-Golas-Hermann:15}).

\item
In Section~\ref{sec:further-look-at-the-categorial-framework},
we have shown that diagram categories
provide a reservoir of \eiu-categories.
Another way to find appropriate categories is the restriction of
\eiu-categories to subcategories. For example, if
one restricts the category $\sigmacat{Graphs}$ to simple graphs, then this
category is closed under empty subobjects, intersections and unions so that
this category inherits all reaction systems over $\sigmacat{Graphs}$ if the
background graph is simple. How do general restriction principles look
like that yield such subcategories?


\item
In Section~\ref{sec:category-reaction-systems-over-cat},
we have shown that monomorphisms on the
background objects provide suitable morphisms between reaction systems
over a category. What about further possibilities?

\item
Another direction of research of this kind may be to consider functors.
For instance, the usual embedding of $\Sigma$-graphs into
$\Sigma$-hypergraphs induces such a functor.
The other way round, the usual
transformation of a hypergraph into a graph can be extended to morphisms.
The question is which properties of a functor $F\colon \cat{C} \to \cat{C’}$
are sufficient such that a reaction system $\mathcal{A}$ over $\cat{C}$
is translated into a reaction system $F(\mathcal{A})$ over $\cat{C’}$.
Whenever this works, one can compare reaction systems
over different categories.

\end{enumerate}

\paragraph{Acknowledgment}
We are grateful to
Nicolas Behr,
Andrea Corradini,
Reiko Heckel,
Berthold Hoffmann,
Jens Kosiol,
Sabine Kuske
and
Gabriele Taenzer
for their valuable hints.
We are particularly grateful
to Grzegorz Rozenberg who encouraged us to investigate graph surfing in
reaction systems from a categorical point of view.
Last but not least, we would like to thank the anonymous reviewers for
their helpful comments and hints that led to some improvements.

\bibliographystyle{eptcs}
\bibliography{lit-rs,lit_all}

\end{document}